\documentclass{amsart}
\RequirePackage[OT1]{fontenc}
\RequirePackage{amsthm,amsmath}
\RequirePackage[numbers]{natbib}
\RequirePackage[colorlinks,citecolor=blue,urlcolor=blue]{hyperref}
\usepackage{graphicx}
\usepackage{latexsym,color}
\usepackage{graphicx}
\usepackage{amssymb}
 \usepackage{amsfonts}
 \usepackage{amsmath}
\usepackage{amsthm}
\newtheorem{theorem}{Theorem}[section]
 
 \newtheorem{lemma}[theorem]{Lemma}
 \newtheorem{remark}[theorem]{Remark}
\newtheorem{corollary}[theorem]{Corollary}

\numberwithin{equation}{section}

\begin{document}

\title{On Shortfall Risk Minimization for Game Options}
\thanks{Partially supported by the Israeli Science Foundation under Grant
  160/17.}
\author{Yan Dolinsky \address{
 Department of Statistics, Hebrew University of Jerusalem, Israel. \\
 e.mail: yan.dolinsky@mail.huji.ac.il}
 ${}$\\
  ${}$\\
Hebrew University
 }

\date{\today}
\begin{abstract}
In this paper we study the existence of an optimal hedging strategy for the shortfall risk
measure in the game options setup. We consider the continuous time Black--Scholes
(BS) model. Our first result says that in the case where the game contingent claim
(GCC) can be exercised only on a finite set of times, there exists an optimal strategy.
Our second and main result is an example which demonstrates that for the case where
the GCC can be stopped on the all time interval, optimal portfolio strategies need not
always exist.
\end{abstract}


\maketitle
\markboth{Y.Dolinsky}{Shortfall Risk Minimization for Game Options}
\renewcommand{\theequation}{\arabic{section}.\arabic{equation}}
\pagenumbering{arabic}

\section{Introduction}\label{sec:1}\setcounter{equation}{0}
A game contingent claim (GCC) or game option, which was introduced in
\cite{Ki1}, is defined
as a contract between the seller and the buyer of the option such
that both have the right to exercise it at any time up to a
maturity date (horizon) $T$. If the buyer exercises the contract
at time $t$ then he receives the payment $Y_t$, but if the seller
exercises (cancels) the contract before the buyer then the latter
receives $X_t$. The difference $\Delta_t=X_t-Y_t$ is the penalty
which the seller pays to the buyer for the contract cancellation.
In short, if the seller will exercise at a stopping time
$\sigma\leq{T}$ and the buyer at a stopping time $\tau\leq{T}$
then the former pays to the latter the amount $H(\sigma,\tau)$
where
\begin{equation}\label{1.1}
H(\sigma,\tau):=X_{\sigma}\mathbb{I}_{\sigma<\tau}+Y_{\tau}\mathbb{I}_{\tau\leq{\sigma}}
\end{equation}
and we set $\mathbb{I}_{Q}=1$ if an event $Q$ occurs and
$\mathbb{I}_{Q}=0$ if not.

A hedge (for the seller) against a GCC is defined here as a pair
$(\pi,\sigma)$
 which
consists of a self--financing strategy
$\pi$
and a stopping time $\sigma$ which is the cancellation time
for the seller. A hedge is called perfect if no matter what exercise time the buyer
chooses, the seller can cover his liability to the buyer (with probability $1$). The
option price is defined as the minimal initial capital which is required for a
perfect hedge. Recall (see \cite{Ki1}) that pricing a GCC in a complete market leads to the value of
a zero sum optimal stopping (Dynkin’s) game under the unique martingale measure.
For additional information about pricing of game options
see, for instance \cite{E,H,HZ,K,KK,KK1}.

In real market conditions an investor (seller) may not be willing for various reasons
to tie in a hedging portfolio the full initial capital required for a perfect hedge. In this
case the seller is ready to accept a risk that his portfolio value at an exercise time may
be less than his obligation to pay and he will need additional funds to fulfill the contract.

We consider the shortfall risk measure which is given by (see \cite{DK1})
$$R(\pi,\sigma):=\sup_{\tau}\mathbb E_{\mathbb P}\left[\left(H(\sigma,\tau)-V^{\pi}_{\sigma\wedge\tau}\right)^{+}\right]$$
where $\{V^{\pi}_t\}_{t=0}^T$ is the wealth process of the portfolio strategy $\pi$ and $\mathbb E_{\mathbb P}$
denotes the expectation
with respect to the market measure.
The supremum is taken over all exercise times of the buyer.

A natural question to ask, is whether for a given initial capital there exists a hedging
strategy which minimizes the shortfall risk (an optimal hedge). For American options the
existence of optimal hedging strategy is proved by applying the Komlós lemma and relies
heavily on the fact that the shortfall risk measure is a convex functional of the wealth
process (see \cite{M,P}).
For the game options setup, the shortfall risk measure, as a functional of the
wealth process is given by
\begin{equation}\label{1.2}
R(\pi):=\inf_{\sigma}
\sup_{\tau}\mathbb E_{\mathbb P}\left[\left(H(\sigma,\tau)-V^{\pi}_{\sigma\wedge\tau}\right)^{+}\right].
\end{equation}
This functional is not necessarily convex (because of the $\inf$) and so the Komlós lemma
can not be applied here.

In this paper we treat the simplest complete, continuous time model, namely the
Black--Scholes (BS) model. Our first result (Theorem \ref{thm2.1}) which is proved in the next
section says that for the case where the option can be exercised only on a finite set of
times, there exists an optimal hedging strategy. The proof is based on the dynamical
programming approach and the randomization technique developed in \cite{R,RR}.
Up to date there are several existence results for risk minimization in the game options setup (see
\cite{DK1,DK2} and Section 5.2 in \cite{Ki2}).
The above papers treat essentially discrete time trading and due to \textit{admissability} conditions
the trading strategies are compact.
In the current setup
trading is done continuously, and so it requires a new method of proof.

In Section 3 we provide the main result of the paper (Theorem \ref{thm3.1}). This is an
example which demonstrates that for the case where the GCC can be stopped on the
whole time interval, optimal portfolio strategies need not always exist.
We combine the
machinery developed in \cite{KW}
with additional ideas which allow to treat the shortfall risk
measure for game options. Formally, we show that the $\inf$ in (\ref{1.2})
which ruins the convexity leads to non existence
of optimal hedging strategies.

\section{Existence Result}
\label{sec:2}\setcounter{equation}{0}
Consider a complete probability space
$(\Omega, \mathcal{F}, \mathbb P)$ together with a standard
one--dimensional Brownian motion
$\{W_t\}_{t=0}^\infty$, and the filtration
$\mathcal{F}_t=\sigma{\{W_s|s\leq{t}\}}$ completed by the null sets.
We consider a simple BS financial market with time horizon $T<\infty$,
which
consists of a riskless savings account bearing zero interest (for simplicity) and of a risky
asset $S$, whose value at time $t$ is given by
\begin{equation*}
S_t=S_0\exp\left(\kappa W_t+(\vartheta-\kappa^2/2) t\right), \ \ t\in [0,T]
\end{equation*}
where $S_0,\kappa>0$ and $\vartheta\in\mathbb R$ are constants.

Define the exponential martingale
\begin{equation}\label{2.1}
Z_t:=\exp\left(-\frac{\vartheta}{\kappa} W_t-\frac{\vartheta^2}{2\kappa^2} t\right),  \ \ t\in [0,T].
\end{equation}
From the Girsanov theorem it follows that the probability measure $\mathbb Q$ which is given by
\begin{equation}\label{2.1+}
\frac{d\mathbb Q}{d\mathbb P} {|\mathcal F_t}:=Z_t, \ \ t\in [0,T]
\end{equation}
is the unique martingale measure for the risky asset $S$.

Next,
let
$\mathbb T:=\{0=T_0<T_1<...<T_n=T\}$ be a finite set of a deterministic times.
Consider a game option
that can be exercised on the set $\mathbb T$. Denote by $\mathcal T_{\mathbb T}$ the set of all stopping times with values
in $\mathbb T$.
For any $k=0,1,...,n$ the payoffs at time $T_k$ are
path--independent and given by
$
Y_{T_k}=f_k(S_{T_k})$ and
$X_{T_k}=g_k(S_{T_k})$
where $f_k,g_k:(0,\infty)\rightarrow \mathbb R$ are measurable functions and $f_k\leq g_k$. The payoff function $H$ is given by (\ref{1.1}).
We will assume the following integrability condition
\begin{equation}\label{2.bound}
\mathbb E_{\mathbb P}[X_{T_k}]<\infty, \ \ k=0,1,...,n.
\end{equation}

A portfolio strategy with an initial capital $x\geq 0$ is a pair $\pi=(x,\gamma)$ such that
$\gamma=\{\gamma_t\}_{t=0}^T$ is a predictable $S$--integrable process and the corresponding wealth process
$$V^\pi_{t}:=x+\int_{0}^t \gamma_u dS_u, \ \ t\in [0,T]$$
satisfies the \textit{admissibility} condition $V^\pi_{t}\geq 0$ a.s. for all $t$.

Let us recall some elementary properties that will be used in the sequel (for details see Chapters IV-V in \cite{RY}).
The continuity of $S$ implies that the wealth process $\{V^\pi_t\}_{t=0}^T$
is continuous as well. Moreover, since $\{S_t\}_{t=0}^T$ is a $\mathbb Q$--martingale then
the wealth process $\{V^\pi_t\}_{t=0}^T$ is a $\mathbb Q$--local martingale, and so from
the \textit{admissibility} condition
 we get that $\{V^{\pi}_t\}_{t=0}^T$
 is a $\mathbb Q$--super martingale. On the other hand, due to the martingale representation theorem,
 for any nonnegative $\mathbb Q$--martingale $\{M_t\}_{t=0}^T$ there exists a portfolio strategy
$\pi$ such that $V^{\pi}_t=M_t$ a.s. for all $t$.

For any $x\geq 0$ denote by $\mathcal A(x)$ the set of all portfolio strategies
with an initial capital $x$.
A hedging strategy with an initial capital $x$ is a pair
$(\pi,\sigma)\in\mathcal A(x)\times\mathcal T_{\mathbb T}$.

The shortfall risk measure is given by
\begin{eqnarray*}
&R_{\mathbb T}(\pi,\sigma):=\sup_{\tau\in\mathcal T_{\mathbb T}}\mathbb E_{\mathbb P}\left[\left(H(\sigma,\tau)-V^{\pi}_{\sigma\wedge\tau}\right)^{+}\right], \ \ (\pi,\sigma)\in\mathcal A(x)\times\mathcal T_{\mathbb T},\\
&R_{\mathbb T}(x):=\inf_{(\pi,\sigma)\in\mathcal A(x)\times \mathcal T_{\mathbb T}}R_{\mathbb T}(\pi,\sigma).
\end{eqnarray*}
Now, we ready to formulate our first result.
\begin{theorem}\label{thm2.1}
For any $x\geq 0$ there exists a hedging strategy $(\hat\pi,\hat\sigma)\in \mathcal A(x)\times \mathcal T_{\mathbb T}$ such that
$$R_{\mathbb T}(\hat\pi,\hat\sigma)=R_{\mathbb T}(x).$$
\end{theorem}
\subsection{Proof of Theorem 2.1}
We start with some preparations.
Let $U:[0,\infty)\times (0,\infty)\rightarrow\mathbb R$ be a measurable function such that for
any $y>0$, $U(\cdot,y)$ is a bounded, nondecreasing and continuous function.
Let $U_c:[0,\infty)\times(0,\infty)\rightarrow\mathbb R$ be the concave envelop of $U$ with respect to the first variable.
Namely, for any $y>0$
the function $U_c(\cdot,y)$ is the minimal concave function which satisfies
$U_c(\cdot,y)\geq U(\cdot,y)$.
Clearly, $U_c$ is continuous in the first variable.
Thus, for any $y>0$ the set $\{x: \ U(x,y)<U_c(x,y)\}$ is open
and so can be written as a countable union of disjoint intervals
\begin{equation}\label{2.4}
\{x: \ U(x,y)<U_c(x,y)\}=\bigcup_{n\in\mathbb N}(a_n(y),b_n(y)).
\end{equation}
From Lemma 2.8 in \cite{R} it follows
 that $U_c(\cdot,y)$ is affine on each of the intervals
 $(a_n(y),b_n(y))$.
 Since $U,U_c$ are continuous in the first variable,
then the functions $a_n,b_n:(0,\infty)\rightarrow\mathbb R_{+}$, $n\in\mathbb N$
are determined by the countable collection of functions $U(q,\cdot),U_c(q,\cdot):
(0,\infty)\rightarrow\mathbb R$ for nonnegative rational $q$,
and so they are measurable.

For any $0\leq t_1<t_2\leq T$ and a $\mathcal F_{t_1}$ measurable random variable
$\Theta_1\geq 0$
denote by $\mathcal H_{t_1,t_2}(\Theta_1)$ the set of all random variables
$\Theta_2\geq 0$ which are $\mathcal{F}_{t_2}$ measurable and satisfy
$\mathbb E_{\mathbb Q}(\Theta_2|\mathcal F_{t_1})\leq \Theta_1.$

The following auxiliary result is an extension of Theorem 5.1 in \cite{R}.
\begin{lemma}\label{lem.main}
 Let $0\leq t_1<t_2\leq T$ and let $\Theta_1\geq 0$ be a $\mathcal F_{t_1}$ measurable random variable.
 For a function $U$ as above, assume that there exists a function $G:\mathbb R\rightarrow\mathbb R$ such that
 $|U(x,y)|\leq G(y)$ for all $x,y$ and $\mathbb E_{\mathbb P}[G(S_{t_2})]<\infty$.
 Then there exists a random variable $\Theta\in \mathcal H_{t_1,t_2}(\Theta_1)$ such that
\begin{eqnarray*}
&\mathbb E_{\mathbb P}\left[U(\Theta,S_{t_2})|\mathcal F_{t_1}\right]\\
&=ess\sup_{\Theta_2\in \mathcal H_{t_1,t_2}(\Theta_1)}\mathbb E_{\mathbb P}\left[U(\Theta_2,S_{t_2})|\mathcal F_{t_1}\right]\\
&=ess\sup_{\Theta_2\in \mathcal H_{t_1,t_2}(\Theta_1)}\mathbb E_{\mathbb P}\left[U_c(\Theta_2,S_{t_2})|\mathcal F_{t_1}\right].
\end{eqnarray*}
\end{lemma}
\begin{proof}
Since $U_c\geq U$, it is sufficient to show that there exists
$\Theta\in \mathcal H_{t_1,t_2}(\Theta_1)$ such that
\begin{eqnarray*}
&\mathbb E_{\mathbb P}\left[U(\Theta,S_{t_2})|\mathcal F_{t_1}\right]\\
&=ess\sup_{\Theta_2\in \mathcal H_{t_1,t_2}(\Theta_1)}\mathbb E_{\mathbb P}\left[U_c(\Theta_2,S_{t_2})|\mathcal F_{t_1}\right].
\end{eqnarray*}
Choose a sequence $\Theta^{(n)}\in \mathcal H_{t_1,t_2}(\Theta_1)$, $n\in\mathbb N$
such that
\begin{equation}\label{2.5}
\lim_{n\rightarrow\infty}\mathbb E_{\mathbb P}\left[U_c(\Theta^{(n)},S_{t_2})|\mathcal F_{t_1}\right]=ess\sup_{\Theta_2\in \mathcal H_{t_1,t_2}(\Theta_1)}\mathbb E_{\mathbb P}\left[U_c(\Theta_2,S_{t_2})|\mathcal F_{t_1}\right].
\end{equation}
From Lemma~A1.1 in
\cite{DelbSch:94} we obtain a sequence
$\Lambda^{(m)}\in conv(\Theta^{(m)},\Theta^{(m+1)},...)$, $m\in\mathbb N$
converging $\mathbb P$ a.s. to a random variable $\Lambda$. The Fatou lemma implies that
$\Lambda\in \mathcal H_{t_1,t_2}(\Theta_1)$.

By applying the dominated convergence theorem, the inequality $|U_c(\cdot,S_{t_2})|\leq G(S_{t_2})$ and the fact that $U_c$ is concave and continuous
in the first variable we obtain
\begin{eqnarray*}
&\mathbb E_{\mathbb P}\left[U_c(\Lambda,S_{t_2})|\mathcal F_{t_1}\right]\\
&=\lim_{n\rightarrow\infty}\mathbb E_{\mathbb P}\left[U_c(\Lambda^{(n)},S_{t_2})|\mathcal F_{t_1}\right]\\
&\geq \lim_{n\rightarrow\infty}\mathbb E_{\mathbb P}\left[U_c(\Lambda^{(n)},S_{t_2})|\mathcal F_{t_1}\right].
\end{eqnarray*}
This together with (\ref{2.5}) gives
\begin{equation}\label{2.6}
\mathbb E_{\mathbb P}\left[U_c(\Lambda,S_{t_2})|\mathcal F_{t_1}\right]=
ess\sup_{\Theta_2\in \mathcal H_{t_1,t_2}(\Theta_1)}\mathbb E_{\mathbb P}\left[U_c(\Theta_2,S_{t_2})|\mathcal F_{t_1}\right].
\end{equation}

Next, introduce the normal random variable
$$\Gamma:=(W_{\frac{t_1+t_2}{2}}-W_{t_1})-\frac{1}{2} (W_{t_2}-W_{t_1}).$$
Observe that $\mathbb E_{\mathbb P}[\Gamma W_{t_2}]=0$ and so we conclude that
$\Gamma$ is independent of the $\sigma$--algebra generated by $W_t, t\in [0,t_1]\cup\{t_2\}$.
From Theorem 1 in \cite{S} it follows that there exists a measurable function
$\Phi:C[0,t_1]\times \mathbb R^2\rightarrow\mathbb R$ such that we have the following equality
of the joint laws
$$\left(\left(W_{[0,t_1]},W_{t_2},\Lambda\right);\mathbb P\right)=\left(\left(W_{[0,t_1]},W_{t_2},\Phi\left(W_{[0,t_1]},W_{t_2},\Gamma\right)\right);\mathbb P\right).$$
In particular $\Phi\left(W_{[0,t_1]},W_{t_2},\Gamma\right)\in \mathcal H_{t_1,t_2}(\Theta_1)$
and
$$\mathbb E_{\mathbb P}\left[U_c\left(\Lambda,S_{t_2}\right)|\mathcal F_{t_1}\right]=\mathbb E_{\mathbb P}\left[U_c\left(\Phi\left(W_{[0,t_1]},W_{t_2},\Gamma\right),S_{t_2}\right)|\mathcal F_{t_1}\right].$$
Thus, without loss of generality we assume that
$\Lambda=\Phi\left(W_{[0,t_1]},W_{t_2},\Gamma\right).$

We arrive to the final step of the proof. Introduce the normal random variable
$$\hat \Gamma:=(W_{\frac{t_1+2t_2}{3}}-W_{t_1})-\frac{2}{3}(W_{t_2}-W_{t_1})-\frac{2}{3}\Gamma.$$
Observe that $\mathbb E_{\mathbb P} [\hat \Gamma W_{t_2}]=\mathbb E_{\mathbb P} [\hat \Gamma \Gamma]=0.$
Thus,
$\hat\Gamma$ is independent of the $\sigma$--algebra generated by $W_t, t\in [0,t_1]\cup\{t_2\}$ and $\Gamma$.
Let $F^{-1}$ be the inverse function of the cumulative distribution function
$F(\cdot):=\mathbb P(\hat\Gamma\leq \cdot)$. Recall (\ref{2.4}) and
define the random variable
\begin{eqnarray*}
&\Theta:=\Lambda\mathbb I_{\Lambda\notin \bigcup_{n\in\mathbb N}(a_n(S_{t_2}),b_n(S_{t_2}))}\\
&+\sum_{n\in\mathbb N}b_n(S_{t_2})\mathbb I_{\Lambda\in (a_n(S_{t_2}),b_n(S_{t_2}))}\mathbb I_{\hat\Gamma<F^{-1}\left(\frac{\Lambda-a_n(S_{t_2})}{b_n(S_{t_2})-a_n(S_{t_2})}\right)}\\
&+\sum_{n\in\mathbb N}a_n(S_{t_2})\mathbb I_{\Lambda\in (a_n(S_{t_2}),b_n(S_{t_2}))}\mathbb I_{\hat\Gamma>F^{-1}\left(\frac{\Lambda-a_n(S_{t_2})}{b_n(S_{t_2})-a_n(S_{t_2})}\right)}.
\end{eqnarray*}
Let $\mathcal G$ be the $\sigma$--algebra generated by $W_t, t\in [0,t_1]\cup\{t_2\}$ and $\Gamma$.
From the Bayes theorem, the tower property for conditional expectation and (\ref{2.1+}) we get
\begin{eqnarray*}
&\mathbb E_{\mathbb Q}\left(\Theta|\mathcal F_{t_1}\right)=\mathbb E_{\mathbb P}\left(\frac{\Theta Z_{t_2}}{Z_{t_1}}|\mathcal F_{t_1}\right)\\
&=\mathbb E_{\mathbb P}\left(\mathbb E_{\mathbb P}\left(\frac{\Theta Z_{t_2}}{Z_{t_1}}|\mathcal G\right)|\mathcal F_{t_1}\right)=
\mathbb E_{\mathbb P}\left(\frac{\Lambda Z_{t_2}}{Z_{t_1}}|\mathcal F_{t_1}\right)=\mathbb E_{\mathbb Q}\left(\Lambda|\mathcal F_{t_1}\right).
\end{eqnarray*}
Thus $\Theta\in\mathcal H_{t_1,t_2}(\Theta_1)$.
Finally, let us notice that
$U(\Theta,S_{t_2})=U_c(\Theta,S_{t_2})$, and so,
from the tower property of conditional expectation and the fact that $U_c(\cdot,y)$ is affine on each of the intervals
 $(a_n(y),b_n(y))$ we obtain
\begin{eqnarray*}
&\mathbb E_{\mathbb P}\left[U(\Theta,S_{t_2})|\mathcal F_{t_1}\right]=
\mathbb E_{\mathbb P}\left[\mathbb E_{\mathbb P}\left(U(\Theta,S_{t_2})|\mathcal G\right)|\mathcal F_{t_1}\right]\\
&=\mathbb E_{\mathbb P}\left[\mathbb E_{\mathbb P}\left(U_c(\Theta,S_{t_2})|\mathcal G\right)|\mathcal F_{t_1}\right]=
\mathbb E_{\mathbb P}\left[U_c(\Lambda,S_{t_2})|\mathcal F_{t_1}\right].
\end{eqnarray*}
This together with
(\ref{2.6}) completes the proof.
\end{proof}
We arrive at the following
Corollary.
\begin{corollary}\label{cor1}
${}$\\
Let $B:[0,\infty)\times(0,\infty)\rightarrow\mathbb R$ be a measurable function such that for
any $y>0$, $B(\cdot,y)$ is a bounded, nonincreasing and continuous function.
Let $B^c:[0,\infty)\times (0,\infty)\rightarrow\mathbb R$ be the convex envelop of $B$ with respect to the first variable.\\
(i). Let $0\leq t_1<t_2\leq T$ and let $\Theta_1\geq 0$ be a $\mathcal F_{t_1}$ measurable random variable.
 Assume that there exists a function $G:\mathbb R\rightarrow\mathbb R$ such that
 $|B(x,y)|\leq G(y)$ for all $x,y$ and $\mathbb E_{\mathbb P}[G(S_{t_2})]<\infty$.
 Then there exists a random variable $\Theta\in \mathcal H_{t_1,t_2}(\Theta_1)$ such that
\begin{eqnarray*}
&\mathbb E_{\mathbb P}\left[B(\Theta,S_{t_2})|\mathcal F_{t_1}\right]\\
&=ess\inf_{\Theta_2\in \mathcal H_{t_1,t_2}(\Theta_1)}\mathbb E_{\mathbb P}\left[B(\Theta_2,S_{t_2})|\mathcal F_{t_1}\right]\\
&=ess\inf_{\Theta_2\in \mathcal H_{t_1,t_2}(\Theta_1)}\mathbb E_{\mathbb P}\left[B^c(\Theta_2,S_{t_2})|\mathcal F_{t_1}\right].
\end{eqnarray*}
(ii). Let $t_1=0$. The function $b:[0,\infty)\rightarrow\mathbb R$ which is defined by
$$b(x):= \inf_{\Theta\in \mathcal H_{0,t_2}(x)}\mathbb E_{\mathbb P}\left[B(\Theta,S_{t_2})\right]=\inf_{\Theta\in \mathcal H_{0,t_2}(x)}\mathbb E_{\mathbb P}\left[B^c(\Theta,S_{t_2})\right], \ \ x\geq 0$$
is convex and continuous.
\end{corollary}
\begin{proof}
${}$\\
(i). The result follows immediately by applying Lemma \ref{lem.main} for $U:=-B$.
\\
(ii). The convexity of $b$ follows from the convexity of $B^c$ in the first variable and the fact that
for any $x_1,x_2\geq 0$ and $\lambda\in (0,1)$,
$$\lambda\mathcal A(x_1)+(1-\lambda)\mathcal A(x_2)\subset \mathcal A(\lambda x_1+(1-\lambda)x_2).$$
In particular $b$ is continuous in $(0,\infty)$. It remans to prove continuity at $x=0$.
Since $B$ is nonincreasing in the first variable, then $b$ is nonincreasing as well. Thus, it is sufficient to show that
$b(0)\leq \lim_{n\rightarrow\infty} b(1/n)$. To that end,
choose
$\Theta^{(n)}\in \mathcal H_{0,t_2}(1/n)$, $n\in\mathbb N$
such that
$$
\lim_{n\rightarrow\infty}\mathbb E_{\mathbb P}\left[B^c(\Theta^{(n)},S_{t_2})\right]=\lim_{n\rightarrow\infty} b(1/n).
$$
From Lemma~A1.1 in
\cite{DelbSch:94} we obtain a sequence
$\Lambda^{(m)}\in conv(\Theta^{(m)},\Theta^{(m+1)},...)$, $m\in\mathbb N$
converging $\mathbb P$ a.s. to a random variable $\Lambda$. The Fatou lemma implies that
$\Lambda=0$, and so by applying the dominated convergence theorem together with convexity and continuity of $B^c$ in the first variable we get,
$$b(0)=\lim_{n\rightarrow\infty}\mathbb E_{\mathbb P}\left[B^c(\Lambda^{(n)},S_{t_2})\right]
\leq\lim_{n\rightarrow\infty}\mathbb E_{\mathbb P}\left[B^c(\Theta^{(n)},S_{t_2})\right]=\lim_{n\rightarrow\infty} b(1/n)$$
and continuity follows.
\end{proof}

Now we are ready to prove Theorem \ref{thm2.1}.
\begin{proof}
Let $x\geq 0$. For any $\pi\in\mathcal A(x)$ we define $R_{\mathbb T}(\pi)$ as in (\ref{1.2}) where the infimum and
the supremum are taken over the set $\mathcal T_{\mathbb T}$.

Moreover, define the random variables $\Psi^{\pi}_k$, $k=0,1...,n$ by
$$
\Psi^{\pi}_n:=\left(Y_T-V^{\pi}_T\right)^{+}
$$
and for $k=0,1,...,n-1$ by the recursive relations
\begin{equation}\label{2.new}
\Psi^{\pi}_k:=
\min \left(\left(X_{T_k}-V^{\pi}_{T_k}\right)^{+},\max\left(\left(Y_{T_k}-V^{\pi}_{T_k}\right)^{+},
\mathbb E_{\mathbb P}(\Psi^{\pi}_{k+1}|\mathcal F_{T_k})\right)\right).
\end{equation}
In view of (\ref{2.bound}) the random variables $\Psi^{\pi}_k$, $k=0,1,...,n$ are well defined.
From the standard theory of zero--sum Dynkin games (see \cite{O})
it follows that
\begin{equation*}
\Psi^{\pi}_0=R_{\mathbb T}(\pi).
\end{equation*}
Moreover, for the stopping time
$$\sigma:=T\wedge\min\left\{t\in\mathbb T:  \ \Psi^{\pi}_{t}=\left(X_t-V^{\pi}_{t}\right)^{+}\right\}$$
we have
$R_{\mathbb T}(\pi)=R_{\mathbb T}(\pi,\sigma)$.

Thus, in order to conclude the proof we need to show that there exists
$\hat\pi\in\mathcal A(x)$ such that
\begin{equation}\label{2.10}
\Psi^{\hat\pi}_0=\inf_{\pi\in\mathcal A(x)}\Psi^{\pi}_0.
\end{equation}
We apply dynamical programming.
Introduce the functions
$B_k:[0,\infty)\times (0,\infty)\rightarrow\mathbb R$, $k=0,1,...,n$ by
$$B_n(z,y):=(f_n(y)-z)^{+},$$
and for $k=0,1,...,n-1$ by the recursive relations
\begin{eqnarray*}
&B_k(z,y)=\min \left(\left(g_k(y)-z\right)^{+},\max\left(\left(f_k(y)-z\right)^{+},\right.\right.\\
&\left.\left.\inf_{\Theta_{k+1}\in \mathcal H_{0,T_{k+1}-T_k}(z)}\mathbb E_{\mathbb P}\left[B_{k+1}(\Theta_{k+1},y S_{T_{k+1}-T_k})\right]\right)\right).
\end{eqnarray*}
Let us argue by backward induction that for any $k$, $B_k(z,y)$ is measurable, and for any $y$ the function
$B_k(\cdot,y)$ is
continuous and nonincreasing. For $k=n$ this is clear. Assume that the statement holds for $k+1$, let us prove it for $k$.
From Corollary \ref{cor1}(ii) it follows that for any $y$ the function
$B_k(\cdot,y)$ is continuous and nonincreasing.
For any $z>0$ the measurability of the function
$B_k(z,\cdot)$ follows from the fact that the
set $\mathcal H_{0,T_{k+1}-T_k}(z)$ is separable
(with respect to convergence in probability).
Since $B_k$ is continuous in the first variable
we conclude joint measurability and
complete the argument.

Next, from Corollary \ref{cor1}(i) it follows that
we can construct a sequence of random variables
$D_0,D_1,...,D_n$ such that
$D_0=x$ and for any $k=1,...,n$
$D_k\in\mathcal H_{T_{k-1},T_k}(D_{k-1})$ satisfies
\begin{equation}\label{2.12}
\mathbb E_{\mathbb P}\left[B_{k}(D_k,S_{T_{k}})|\mathcal F_{T_{k-1}}\right]=
ess\inf_{\Theta_{k}\in \mathcal H_{T_{k-1},T_{k}}(D_{k-1})}\
\mathbb E_{\mathbb P}\left[B_{k}(\Theta_k,S_{T_{k}})|\mathcal F_{T_{k-1}}\right].
\end{equation}
Since $B_k$, $k=0,1,...,n$  are nonincreasing in the first variable then without loss of generality we assume that
$\mathbb E_{\mathbb Q}[D_k|\mathcal F_{T_{k-1}}]=D_{k-1}$ for all $k$.

Finally, the completeness of the BS model implies that there exists $\hat\pi\in\mathcal A(x)$ such that
$V^{\hat\pi}_{T_k}=D_k$ for all $k=0,1,...,n$.
Observe that $\frac{S_{T_k}}{S_{T_{k-1}}}$ is independent of
$\mathcal F_{T_{k-1}}$ and has the same distribution as $S_{T_{k}-T_{k-1}}$. Thus, from
(\ref{2.new})
 and (\ref{2.12}) we obtain (by backward induction)
\begin{equation}\label{2.13}
B_k(V^{\hat\pi}_{T_k}, S_{T_k})= \Psi^{\hat\pi}_k \ \ \mbox{a.s}. \ \  \forall k=0,1,...,n.
\end{equation}
On the other hand,
for an arbitrary
$\pi\in\mathcal A(x)$  we have
$V^{\pi}_{T_k}\in \mathcal H_{T_{k-1},T_k}(V^{\pi}_{T_{k-1}})$, $k=1,...,n$.
Hence, similar arguments as before (\ref{2.13}) yield
\begin{equation}\label{2.11}
B_k(V^{\pi}_{T_k}, S_{T_k})\leq \Psi^{\pi}_k \ \ \mbox{a.s}. \ \  \forall k=0,1,...,n.
\end{equation}

By combining (\ref{2.13})--(\ref{2.11}) for $k=0$ gives that for any $\pi\in\mathcal A(x)$
$$\Psi^{\hat\pi}_0=B_0(x,S_0)\leq \Psi^{\pi}_0$$
and (\ref{2.10}) follows.
\end{proof}
\begin{remark}
We observe that
the proof of Theorem \ref{thm2.1} and Lemma \ref{lem.main}
can be adjusted to the case where the volatility and the drift are deterministic functions
of time.
However, in order to make the presentation more friendly we assume constant parameters.
\end{remark}

\section{Example Where no Optimal Strategy Exists}
\label{sec:3}\setcounter{equation}{0}
In this section we consider a game option which can be exercised at any time
in the interval $[0,1]$.
The payoffs are given by
\begin{eqnarray*}
&X_t=(1+\sin (\pi t))\max(Z_t,1/2), \ \ t\in [0,1]\\
&Y_1=X_1,\\
&Y_t=0, \ \  \mbox{for} \ \ t<1
\end{eqnarray*}
where $Z_t$ was defined in (\ref{2.1}).
Notice that $\mathbb E_{\mathbb P}[\sup_{0\leq t\leq 1} X_t]<\infty$.

Denote by $\mathcal T$ the set of all stopping times with values in the interval $[0,1]$.
Obviously, the equalities $Y_{[0,1)}\equiv 0$ and $Y_1=X_1$ imply that the shortfall risk measure is given by
\begin{eqnarray*}
&R(\pi,\sigma)=\mathbb E_{\mathbb P}\left[\left(X_{\sigma}-V^{\pi}_{\sigma}\right)^{+}\right]\\
&R(\pi)=\inf_{\sigma\in\mathcal T}\mathbb E_{\mathbb P}\left[\left(X_{\sigma}-V^{\pi}_{\sigma}\right)^{+}\right]\\
&R(x)=\inf_{\sigma\in\mathcal T}\inf_{\pi\in\mathcal A(x)}
\mathbb E_{\mathbb P}\left[\left(X_{\sigma}-V^{\pi}_{\sigma}\right)^{+}\right].
\end{eqnarray*}
For any $\pi$ the process
$\{(X_t-V^{\pi}_t)^{+}\}_{t=0}^1$ is continuous, and so from the general theory of optimal stopping (see Section 6 in \cite{Ki2})
it follows that there exists $\sigma=\sigma(\pi)$ such that $R(\pi,\sigma)=R(\pi)$.
Namely, the existence of an optimal hedging strategy is equivalent to the existence of an optimal portfolio strategy.
We say that $\pi\in\mathcal A(x)$ is an optimal portfolio strategy if $R(\pi)=R(x)$.

We arrive at the main result.
\begin{theorem}\label{thm3.1}
Assume that the drift term $\vartheta\neq 0$ and let
\begin{equation}\label{3.nu}
\nu:=\frac{1}{2}\mathbb E_{\mathbb P}\left[Z_1\mathbb{I}_{Z_1<1/2}\right],
\end{equation}
observe that $\vartheta\neq 0$ implies that $\nu>0$.
Then for any initial capital $x\in (0,\nu)$ there is no optimal strategy.
\end{theorem}
\begin{remark}
If $\vartheta=0$ then $\mathbb P=\mathbb Q$.
In this specific case (see Theorem 7.1 in \cite{DK1})
there exists an optimal hedging strategy.
\end{remark}
Before we prove Theorem \ref{thm3.1} we will need some auxiliary results.
We start with the following lemma.
\begin{lemma}\label{lem3.1}
The function $R:[0,\infty)\rightarrow [0,\infty)$ is convex and continuous.
\end{lemma}
\begin{proof}
The proof will be done by approximating $R(\cdot)$.
For any $n\in\mathbb N$ let $\mathcal T_n$
be set of all stopping times with values in the set
$\{1/n,2/n,...,1\}$ ($0$ is not included).
Set,
\begin{eqnarray*}
&R_n(\pi):=\inf_{\sigma\in\mathcal T_n}R(\pi,\sigma)\\
&R_n(x):=\inf_{\sigma\in\mathcal T_n}\inf_{\pi\in\mathcal A(x)}
R(\pi,\sigma).
\end{eqnarray*}
We argue that $R_n$ converge uniformly to $R$. First, we have the obvious observation $R_n(\cdot)\geq R(\cdot)$.
Next, let $x\geq 0$ and $(\pi,\sigma)\in\mathcal A(x)\times\mathcal T$. Define
 $\sigma_n\in\mathcal T_n$ by
 $$\sigma_n:=\frac{1}{n}\min\{k\in\mathbb N: \ k/n\geq \sigma\}.$$ Clearly, $\sigma_n\geq \sigma$.
 Thus, there exists a portfolio $\pi_n\in\mathcal A(x)$ such that
 $V^{\pi_n}_{\sigma_n}=V^{\pi}_{\sigma}$.
 From the inequality $\sigma_n-\sigma\leq 1/n$ we obtain
 $$R(\pi_n,\sigma_n)-R(\pi,\sigma)\leq \mathbb E_{\mathbb P}\left[|X_{\sigma}-X_{\sigma_n}|\right]\leq \mathbb E_{\mathbb P}\left[\sup_{|t-s|\leq 1/n}|X_t-X_s|\right].$$
 Since $(\pi,\sigma)\in\mathcal A(x)\times\mathcal T$ was arbitrary we conclude that
 $$0\leq R_n(x)-R(x)\leq \mathbb E_{\mathbb P}\left[\sup_{|t-s|\leq 1/n}|X_t-X_s|\right].$$ From the dominated convergence theorem
 $$\lim_{n\rightarrow\infty}\mathbb E_{\mathbb P}\left[\sup_{|t-s|\leq 1/n}|X_t-X_s|\right]=0$$ and uniform convergence
follows.

It remains to argue that for any $n$ the function $R_n:[0,\infty)\rightarrow [0,\infty)$ is convex and continuous. Fix $n\in\mathbb N$.
For any $k=1,...,n$ let
 $g_k:(0,\infty)\rightarrow (0,\infty)$ be such that $X_{k/n}=g_k(S_{k/n})$.
Introduce the functions
$\hat B_k:[0,\infty)\times (0,\infty)\rightarrow\mathbb R$, $k=0,1,...,n$ by
$$\hat B_n(z,y):=(g_n(y)-z)^{+},$$
for $k=1,...,n-1$ by the recursive relations
$$\hat B_k(z,y)=\min \left(\left(g_k(y)-z\right)^{+},
\inf_{\Theta_{k+1}\in \mathcal H_{0,1/n}(z)}\mathbb E_{\mathbb P}\left[\hat B_{k+1}(\Theta_{k+1},y S_{1/n})\right]\right),
$$
and for $k=0$
\begin{equation}\label{3.new}
\hat B_0(z,y)=
\inf_{\Theta_{1}\in \mathcal H_{0,1/n}(z)}\mathbb E_{\mathbb P}\left[\hat B_{1}(\Theta_{1},y S_{1/n})\right].
\end{equation}
Observe that $R_n(\cdot)$ is "almost" as $R_{\mathbb T}(\cdot)$ defined in Section 2 for
 the set $\mathbb T:=\{0,1/n,2/n,...,1\}$, the only difference is that for $R_n(x)$ stopping at zero is not allowed.
 This is why in (\ref{3.new}) we do not take minimum with $(g_0(y)-z)^{+}$.
 Using similar arguments as in the proof of Theorem \ref{thm2.1} we obtain that
$R_n(x)=\hat B_0(x,S_0)$.
 Finally, from Corollary \ref{cor1}(ii) we get that for any $y$,
 $\hat B_0(\cdot,y)$ is convex and continuous. This completes the proof.
 \end{proof}

Next, we
observe that for any stopping time $\sigma\in\mathcal T$ and $\lambda>0$
\begin{equation}\label{3.main}
\inf_{\Upsilon\geq 0} \left[(X_{\sigma}-\Upsilon)^{+}+\lambda Z_{\sigma} \Upsilon\right]=X_{\sigma}\min(1,\lambda Z_{\sigma}).
\end{equation}
This brings us to introducing the
function
\begin{equation}\label{3.dual}
F(\lambda)=\inf_{\sigma\in\mathcal T}\mathbb E_{\mathbb P}\left[X_{\sigma}\min(1,\lambda Z_{\sigma})\right], \ \ \lambda>0.
\end{equation}
Obviously $F:(0,\infty)\rightarrow [0,\infty)$ is concave and nondecreasing.
Inspired by Corollary 8.3 in \cite{KW} we prove the following.
\begin{lemma}\label{lem3.2}
${}$\\
(i). For any $x\geq 0$ and $\lambda>0$,
 $$R(x)\geq F(\lambda)-\lambda x.$$
(ii). Let $\lambda>0$ be such that $F$ is
differentiable at $\lambda$. Then for $x=F'(\lambda)$ we have the equality
$$R(x)=F(\lambda)-\lambda x.$$
\end{lemma}
\begin{proof}
${}$\\
(i). Let $x\geq 0$ and $\lambda>0$. Choose arbitrary $(\pi,\sigma)\in\mathcal A(x)\in\mathcal T$. Then, from the super--martingale property
of an \textit{admissible} portfolio we have
\begin{equation}\label{3.admis}
x=V^{\pi}_0\geq\mathbb E_{\mathbb Q}[V^{\pi}_{\sigma}]=\mathbb E_{\mathbb P}[Z_{\sigma}V^{\pi}_{\sigma}].
\end{equation}
This together with (\ref{3.main}) gives
$$R(\pi,\sigma)+\lambda x \geq\mathbb E_{\mathbb P}\left[(X_{\sigma}-V^{\pi}_{\sigma})^{+}+\lambda Z_{\sigma} V^{\pi}_{\sigma}\right]\geq
F(\lambda).$$
Since $(\pi,\sigma)\in\mathcal A(x)\in\mathcal T$ was arbitrary we complete the proof.\\
(ii). In view of (i), it is sufficient to show that
$R(x)\leq F(\lambda)-\lambda x$.
Let $\sigma_{\lambda}\in\mathcal T$ be an optimal stopping time in (\ref{3.dual}), i.e.
\begin{equation}\label{3.5}
F(\lambda)=\mathbb E_{\mathbb P}\left[X_{\sigma_{\lambda}}\min(1,\lambda Z_{\sigma_{\lambda}})\right].
\end{equation}
Such stopping time exists because the process $\{X_{t}\min(1,\lambda Z_{t})\}_{t=0}^1$ is continuous.
Set $\Upsilon_{\lambda}=X_{\sigma_{\lambda}}\mathbb I_{Z_{\sigma_{\lambda}}<1/\lambda}$.
From (\ref{3.main}) it follows that for any $\tilde\lambda>0$
$$
F(\tilde\lambda)\leq \mathbb E_{\mathbb P}\left[(X_{\sigma_{\lambda}}-\Upsilon_{\lambda})^{+}+\tilde\lambda Z_{\sigma_{\lambda}} \Upsilon_{\lambda}\right].$$
On the other hand from (\ref{3.5})
$$
F(\lambda)= \mathbb E_{\mathbb P}\left[(X_{\sigma_{\lambda}}-\Upsilon_{\lambda})^{+}+\lambda Z_{\sigma_{\lambda}} \Upsilon_{\lambda}\right].$$
Thus,
\begin{eqnarray*}
&\frac{F(\tilde\lambda)-F(\lambda)}{\tilde\lambda-\lambda}\leq \mathbb E_{\mathbb P}\left[ Z_{\sigma_{\lambda}} \Upsilon_{\lambda}\right], \ \ \mbox{for} \ \ \tilde\lambda>\lambda\\
&\mbox{and} \ \ \frac{F(\tilde\lambda)-F(\lambda)}{\tilde\lambda-\lambda}\geq \mathbb E_{\mathbb P}\left[ Z_{\sigma_{\lambda}} \Upsilon_{\lambda}\right]\ \ \mbox{for} \ \ \tilde\lambda<\lambda.
\end{eqnarray*}
From the fact that $F'(\lambda)=x$ we conclude that
$$x=\mathbb E_{\mathbb P}\left[ Z_{\sigma_{\lambda}} \Upsilon_{\lambda}\right]=\mathbb E_{\mathbb Q}\left[\Upsilon_{\lambda}\right].$$
The completeness of the BS model implies that there exists $\pi\in\mathcal A(x)$ such that
$V^{\pi}_{\sigma_{\lambda}}=\Upsilon_{\lambda}$. From (\ref{3.5})
we get
\begin{eqnarray*}
&R(x)+\lambda x\leq R(\pi,\sigma_{\lambda})+\lambda x=\mathbb E_{\mathbb P}\left[(X_{\sigma_{\lambda}}-\Upsilon_{\lambda})^{+}+\lambda Z_{\sigma_{\lambda}} \Upsilon_{\lambda}\right]\\
&=\mathbb E_{\mathbb P}\left[X_{\sigma_{\lambda}}\min(1,\lambda Z_{\sigma_{\lambda}})\right]=F(\lambda)
\end{eqnarray*}
as required.
\end{proof}
While Lemma \ref{lem3.1}--\ref{lem3.2} are quite general, the following lemma uses the explicit structure of the payoff process $\{X_t\}_{t=0}^1$.
\begin{lemma}\label{lem3.2+}
${}$\\
(i).
For any $\lambda\geq 2$, $F(\lambda)=1$. \\
(ii). The derivative of $F$ from the left (exists because $F$ is concave) satisfies
$F'_{-}(2)\geq \nu$ where
$\nu$ is given by (\ref{3.nu}).
\end{lemma}
\begin{proof}
${}$\\
(i).
Let $\lambda\geq 2$. Obviously, $\mathbb E_{\mathbb P}[Z_{\sigma}]=1$ for all $\sigma\in\mathcal T$. Hence, from the simple formula
$\max(z,1/2)\min(1,2 z)\equiv z$ we obtain
$$F(\lambda)\geq F(2)=\inf_{\sigma\in\mathcal T}\mathbb E_{\mathbb P}\left[Z_{\sigma}(1+\sin(\pi\sigma))\right]\geq 1.$$
On the other hand, taking $\sigma\equiv 0$ in (\ref{3.dual}), we get
$F(\lambda)\leq 1$ and so $F\equiv 1$ on the interval $[2,\infty)$.\\
(ii).
Choose $\lambda<2$. Clearly, (we take $\sigma\equiv 1$ in (\ref{3.dual}))
\begin{eqnarray*}
&F(\lambda)\leq\mathbb E_{\mathbb P}\left[\max(Z_1,1/2)\min(1,\lambda Z_{1})\right]\\
&\leq\mathbb E_{\mathbb P}\left[Z_1\mathbb I_{Z_1>1/2}+\frac{\lambda}{2}Z_1\mathbb I_{Z_1<1/2}\right]=1-\frac{2-\lambda}{2}
\mathbb E_{\mathbb P}\left[Z_1\mathbb I_{Z_1<1/2}\right].
\end{eqnarray*}
This together with the equality $F(2)=1$ gives
$F'_{-}(2)\geq \nu.$
 \end{proof}

Now, we have all the ingredients for the proof of Theorem \ref{thm3.1}.
\begin{proof}
From Lemma \ref{lem3.2}(i) and Lemma \ref{lem3.2+}(i) it follows that for any $x$
$$R(x)\geq F(2)-2 x=1-2x.$$
Let us prove that
\begin{equation}\label{3.100}
R(x)= 1-2x, \ \ \forall x\leq F'_{-}(2).
\end{equation}
Since $R$ is convex (Lemma \ref{lem3.1})
then it is sufficient to show that
$R(0)\leq 1$ and
$R(F'_{-}(2))\leq 1-2 F'_{-}(2)$.

The first inequality is trivial, $R(0)\leq X_0=1$.
Let us show the second inequality.
The concavity of $F$ implies that
there exists a sequence $\lambda_n\uparrow 2$ such that for any $n$ the derivative $F'(\lambda_n)$ exists.
Hence, from the continuity of $R$ (Lemma \ref{lem3.1}), the concavity of $F$ and Lemma \ref{lem3.2}(ii) we obtain
$$R(F'_{-}(2))=\lim_{n\rightarrow\infty}R(F'(\lambda_n))=\lim_{n\rightarrow\infty}[F(\lambda_n)-\lambda_n F'(\lambda_n)]=1-2 F'_{-}(2)$$
and (\ref{3.100}) follows.

Next, let $x\in (0,\nu)$. Assume by contradiction that there exists a hedging strategy $(\pi,\sigma)\in\mathcal A(x)\times\mathcal T$ such that
$R(\pi,\sigma)=R(x)$. From Lemma \ref{lem3.2+}(ii) and (\ref{3.100})
we obtain
\begin{equation}\label{3.101}
R(\pi,\sigma)=1-2x.
\end{equation}
Observe that if $\sigma$ takes on values (with positive probability) in the interval $(0,1)$ then
$$\mathbb E_{\mathbb P}\left[X_{\sigma}\min(1,2 Z_{\sigma})\right]=\mathbb E_{\mathbb P}\left[Z_{\sigma}(1+\sin(\pi\sigma))\right]>
\mathbb E_{\mathbb P}[Z_{\sigma}]=1.$$
 Thus, from (\ref{3.main}) and (\ref{3.admis})
 $$R(\pi,\sigma)+2  x \geq\mathbb E_{\mathbb P}\left[(X_{\sigma}-V^{\pi}_{\sigma})^{+}+2 Z_{\sigma} V^{\pi}_{\sigma}\right]>1
 $$
which is a contradiction to (\ref{3.101}).
On the other hand if $\sigma \equiv 0$ then
$$R(\pi,\sigma)=X_0-x=1-x,$$
also a contradiction to (\ref{3.101}).

We conclude that the only remaining possibility is $\sigma\equiv 1$. Let us show that there is a contradiction in this case as well.
Introduce the event $$A:=\{\max(Z_1,V^{\pi}_1)<1/2\}.$$
Observe that on the event $A$ we have
\begin{equation}\label{3.final}
(X_{1}-V^{\pi}_1)^{+}+2 Z_{1} V^{\pi}_1=(1/2-Z_1)(1-2V^{\pi}_1)+Z_1 > Z_1=X_1\min(1,2Z_1).
\end{equation}
From (\ref{3.admis}) and the fact that $x<\nu$ it follows that
$$\mathbb E_{\mathbb P}[Z_1V^{\pi}_1]<\nu=\frac{1}{2} \mathbb E_{\mathbb P}\left[Z_1\mathbb I_{Z_1<1/2}\right].$$
This together with the inequality $V^{\pi}_1\geq 0$ gives $\mathbb P(A)>0$.
Thus,
by combining (\ref{3.main}), (\ref{3.admis}) and (\ref{3.final}) we obtain
\begin{eqnarray*}
&R(\pi,\sigma)+2 x=\mathbb E_{\mathbb P}\left[(X_{1}-V^{\pi}_1)^{+}+2 Z_{1} V^{\pi}_1\right]\nonumber\\
&>\mathbb E_{\mathbb P}\left[X_{1}\min(1,2 Z_{1})\right]=\mathbb E_{\mathbb P}[Z_1]=1
\end{eqnarray*}
which is a contradiction to (\ref{3.101}).
\end{proof}
\begin{remark}
The message of Theorem \ref{thm3.1} is that the $\inf$ in (\ref{1.2}) which ruins the convexity of the shortfall risk functional
$R(\pi)$ can lead to non existence of an optimal strategy.
Observe that in the above constructed
example, the payoff process $X$ is continuous 
and the payoff process
$Y$ has a positive jump in the maturity date.

One can ask, what if we require
that both of the payoff processes $X$ and $Y$ will be continuous, is there a counter example in this case as well?

The answer is yes. Let us 
apply Theorem \ref{thm3.1} in order to construct a counter example with continuous payoffs.

Consider a simple BS financial market with time horizon $T=2$
which
consists of a riskless savings account bearing zero interest and of a risky
asset $S$, whose value at time $t$ is given by
\begin{eqnarray*}
&S_t=S_0\exp\left(\kappa W_t+(\vartheta-\kappa^2/2) t\right), \ \ t\in [0,1]\\
&S_t=S_1\exp\left(\kappa (W_t-W_1)-\kappa^2(t-1)/2\right), \ \  t\in (1,2]
\end{eqnarray*}
where, as before,  $S_0,\kappa>0$ and $\vartheta\neq 0$ are constants.
Namely, this is a BS model which has a drift jump in $t=1$. Obviously this market is complete and the unique martingale measure is given by
$\frac{d\mathbb Q}{d\mathbb P} {|\mathcal F_t}:=Z_{t\wedge 1}$
where $Z_t$ is given by (\ref{2.1}). Consider a game option with the continuous payoffs
\begin{eqnarray*}
&\hat X_t=(1+\sin (\pi t))\max(Z_t,1/2), \ \ t\in [0,1]\\
&\hat X_t=\hat X_1,  \ \ t\in (1,2],\\
&\hat Y_t=0, \ \  t\in [0,1]\\
&\hat Y_t= (t-1) \hat X_1 , \ \ t\in (1,2].
\end{eqnarray*}
Denote by $\hat R$ the corresponding shortfall risk.
We argue that for an initial capital $0<x<\nu:=\frac{1}{2} \mathbb E_{\mathbb P}\left[Z_1\mathbb I_{Z_1<1/2}\right]$ 
there is
no optimal hedging strategy.   

Indeed, let $\pi$ be an \textit{admissible} portfolio strategy and $\sigma$ be a stopping time
with values in the interval $[0,2]$.
From the super--martingale property of the portfolio value and the fact that $Z$ is a constant random variable
after $t=1$ we obtain
 $$V^{\pi}_{\sigma\wedge 1}\geq\mathbb E_{\mathbb P}[V^{\pi}_{\sigma}|\mathcal F_1].$$
 This together with the Jensen inequality and the fact that $\hat X$ is a constant random variable after $t=1$ gives
 \begin{equation}\label{3.rem}
 \mathbb E_{\mathbb P}\left[(\hat X_{\sigma\wedge 1}-V^{\pi}_{\sigma\wedge 1})^{+}\right]\leq
 \mathbb E_{\mathbb P}\left[\mathbb E_{\mathbb P}\left[(\hat X_{\sigma}-V^{\pi}_{\sigma})^{+}|\mathcal F_1\right]\right]
 =\mathbb E_{\mathbb P}\left[(\hat X_{\sigma}-V^{\pi}_{\sigma})^{+}\right].
 \end{equation}
  From (\ref{3.rem}), and the relations 
$\hat Y_{[0,1]}\equiv 0$, $\hat Y_2=\hat X_2$ we obtain 
$$
\hat R(\pi,\sigma\wedge 1)=\mathbb E_{\mathbb P}\left[(\hat X_{\sigma\wedge 1}-V^{\pi}_{\sigma\wedge 1})^{+}\right]
\leq \mathbb E_{\mathbb P}\left[(\hat X_{\sigma}-V^{\pi}_{\sigma})^{+}\right]\leq \hat R(\pi,\sigma).
$$
Namely, we can restrict the investor to stopping times in the interval $[0,1]$, 
but this is exactly the setup that was studied in Theorem \ref{thm3.1}.
From Theorem \ref{thm3.1} we conclude that there is no optimal hedging strategy for $x\in (0,\nu).$
\end{remark}
\section*{Acknowledgments}
I would like to thank Yuri Kifer
for introducing me to the problems which are treated
in this paper and also for related fruitful discussions.
I also would like to thank Walter Schachermayer for sharing
some ideas a while ago which turned out to be helpful for proving Theorem \ref{thm2.1}.

{}

\end{document}